\pdfoutput=1
\documentclass[american,aps,pra,reprint, superscriptaddress]{revtex4-1}
\usepackage[unicode=true,pdfusetitle, bookmarks=true,bookmarksnumbered=false,bookmarksopen=false, breaklinks=false,pdfborder={0 0 0},backref=false,colorlinks=false] {hyperref}
\hypersetup{ colorlinks,linkcolor=myurlcolor,citecolor=myurlcolor,urlcolor=myurlcolor}
\usepackage{graphics,epstopdf,graphicx, amsthm, amsmath, amssymb, times, braket, color, bm}
\usepackage[up]{subfigure}
\usepackage{cleveref}

\definecolor{myurlcolor}{rgb}{0,0,0.7}

\newcommand{\proj}[1]{| #1\rangle\!\langle #1 |}

\DeclareMathOperator{\trace}{Tr}
\newcommand{\Ptr}[2]{\trace_{#1}\Pa{#2}}
\newcommand{\Tr}[1]{\Ptr{}{#1}}
\newcommand{\Innerm}[3]{\left\langle #1 \left| #2 \right| #3 \right\rangle}
\newcommand{\Pa}[1]{\left[#1\right]}

\newcommand{\norm}[1]{\left\lVert #1 \right\rVert}

\theoremstyle{plain}
\newtheorem{thm}{\protect\theoremname}
\newtheorem{prop}[thm]{Proposition}
\newtheorem{lem}[thm]{Lemma}
\providecommand{\theoremname}{Theorem}
\newcommand*{\myproofname}{Proof}
\newenvironment{mproof}[1][\myproofname]{\begin{proof}[#1]}{\end{proof}}

\def\ot{\otimes}

\def\real{\mathbb{R}}

\def\cI{\mathcal{I}}

\def\cD{\mathcal{D}}

\def\cH{\mathcal{H}}

\def\cE{\mathcal{E}}

\makeatother

\begin{document}

  \author{Kaifeng Bu}
 \email{kfbu@fas.harvard.edu}
 \affiliation{School of Mathematical Sciences, Zhejiang University, Hangzhou 310027, PR~China}
 \affiliation{Department of Physics, Harvard University, Cambridge, MA 02138, USA}

  \author{Lu Li}
  \email{lilu93@zju.edu.cn}
 \affiliation{School of Mathematical Sciences, Zhejiang University, Hangzhou 310027, PR~China}

   \author{Arun Kumar Pati}
  \email{akpati@hri.res.in}
\affiliation{Quantum Information and Computation Group, Harish-Chandra Research Institute, HBNI, Allahabad, 211019, India}

    \author{Shao-Ming Fei}
    \email{feishm@cnu.edu.cn}
 \affiliation{School of Mathematical Sciences, Capital Normal University, Beijing 100048, PR China}
  \affiliation{Max-Planck-Institute for Mathematics in the Sciences, 04103 Leipzig, Germany}

 \author{Junde Wu}
 \email{wjd@zju.edu.cn}
 \affiliation{School of Mathematical Sciences, Zhejiang University, Hangzhou 310027, PR~China}

\title{Distribution of coherence in multipartite systems under entropic coherence measure}

\begin{abstract}
The distribution of coherence in multipartite systems is one of the fundamental problems in the resource theory of coherence. To quantify the coherence in
multipartite systems more precisely, we introduce new coherence measures, incoherent-quantum (IQ) coherence measures, on bipartite systems by the max- and min- relative entropies and provide the
 operational interpretation  in certain subchannel discrimination problem. By introducing the
smooth max- and min- relative entropies of  incoherent-quantum (IQ) coherence on bipartite systems, we exhibit the distribution of coherence in multipartite systems: the total coherence
is lower bounded by the sum of local coherence and genuine multipartite entanglement.  Besides, we find the monogamy relationship for coherence on multipartite systems by incoherent-quantum (IQ) coherence measures.  Thus, the IQ coherence measures  introduced here
truly capture the non-sharability of quantumness of coherence in multipartite context.
\end{abstract}

\maketitle

\section{Introduction}
The key feature of quantumness in a single system can be
captured by quantum coherence, stemming from the superposition principle in quantum mechanics.
Quantum coherence, as one of the most primitive quantum resource, plays a crucial role in  a variety of applications
ranging from thermodynamics \cite{Lostaglio2015,Lostaglio2015NC} to metrology \cite{Giovannetti2011}.
Recently, the resource theory of coherence has attracted much attention \cite{Baumgratz2014,Girolami2014,Streltsov2015,Winter2016,Killoran2016,Chitambar2016,Chitambar2016a}.
There are other notable resource theories including quantum entanglement \cite{HorodeckiRMP09}, asymmetry \cite{Bartlett2007,Gour2008,Gour2009,Marvian2012,Marvian2013,Marvian14,Marvian2014}, thermodynamics \cite{Fernando2013}, and steering \cite{Rodrigo2015},
among which entanglement is the most famous one and can be used as a basic resource for various quantum information processing protocols such as
superdense coding \cite{Bennett1992}, remote state preparation \cite{Pati2000,Bennett2001} and quantum teleportation \cite{Bennett1993}.

In a resource theory, there are two basic elements:
free states and free operations.
The free states in the resource theory of coherence are called incoherent states, which are defined
as the diagonal states in a given reference basis $\set{\ket{i}}^{d-1}_{i=0}$ for a d-dimensional system. The set of
incoherent states is denoted by $\cI$. Any quantum state can be mapped to an incoherent state by the full dephasing
operation $\Delta(\rho)=\sum^{d-1}_{i=0}\Innerm{i}{\rho}{i}\proj{i}$. However, there is still no general consensus on the set of free operations
in the resource theory of coherence. Here, we take the incoherent operations (IO)  as the free operations, where an operation $\Lambda$ is called an incoherent operation (IO) if there exists a set of Kraus operators $\{K_i\}$ of $\Lambda$  such that $K_i\mathcal{I} K^{\dag}_i\subseteq \mathcal{I}$ for each $i$ \cite{Baumgratz2014}.
To quantify the amount of coherence in the states, several operational coherence measures have been proposed, namely,
the relative entropy of coherence \cite{Baumgratz2014}, the $l_1$ norm of coherence
\cite{Baumgratz2014}, the max-relative entropy of coherence \cite{Bu2017b,Chitambar2016b}  and the robustness of coherence \cite{Napoli2016}.
 These coherence measures provide the lucid quantitative and operational
description of coherence.

The distribution of quantum correlations in multipartite systems is one of
the fundamental properties distinguishing quantum correlations from the classical
ones, as quantum correlations cannot be shared freely by the subsystems.
For example, for any pure tripartite state,  if  Alice and Bob share a maximally entangled state, then
neither Alice nor Bob can be entangled with Charlie, which is dubbed as the monogamy of entanglement \cite{Coffman2000,Koashi2004,Osborne2006,Winter2016Entdis}.
Besides, the monogamy of coherence has been investigated in Refs. \cite{Yao2015,Kumar2017}, where it has been shown that
the monogamy of coherence for relative entropy of coherence does not hold in general.
The distribution of coherence in bipartite and multipartite systems  has also been investigated in
Ref. \cite{Ma2016} and \cite{Radhakrishnan2016}, respectively.
In \cite{Radhakrishnan2016}, the following trade-off relation in multipartite systems has been demonstrated,
\begin{eqnarray}
C\leq C_I+C_L,
\end{eqnarray}
where $C$  is the total coherence of the whole system,  $C_I$ is called intrinsic coherence which captures the coherence between different subsystems, $C_L$
is called local coherence which describes the coherence located on each subsystem, and all these three coherence measures are defined by
some distance measure which is required to satisfy some conditions, such as  triangle inequality. However, it seems that none of the coherence measures defined by  $l_1$ norm, relative entropy  or
Jensen-Shannon divergence meets the requirements in Ref. \cite{Radhakrishnan2016}, as  the existence of triangle inequality for relative entropy  and
Jensen-Shannon divergence is still unknown, while the $l_1$ norm is superadditive for product states, for which a coherence measure is required to be subadditive so that the local coherence
will be upper bounded by  the sum of  the coherence in each subsystem \cite{Radhakrishnan2016}.
Thus, a rigorous characterization of the distribution of coherence in multipartite system
is imperative and of paramount importance.

Here, we investigate the distribution of coherence in multipartite system in terms of the max- and min- relative entropies.
The well-known conditional and unconditional max- and min- entropies  \cite{Renner2004IEEE,Renner2005phd}
can be derived from the max- and min- relative entropies.
Max- and min-relative entropies have  also  been used to define entanglement monotones and their operational significance in manipulation of entanglement has also been
provided in Refs. \cite{Datta2009IEEE,Datta2009,Brandao2011,Buscemi2010}. Coherence measures  based on the max- and min-relative entropies have been introduced in Refs. \cite{Chitambar2016b,Bu2017b}, where
the operational interpretations have also been provided in Ref. \cite{Bu2017b}. In this letter, incoherent-quantum (IQ) coherence measures on
bipartite systems are introduced in terms of the max- and min-relative entropies, which capture the maximal advantage of  bipartite states in certain subchannel discrimination problems. By introducing the smooth max- and min-relative entropies of IQ coherence measures on bipartite systems,
we find that the total coherence of a multipartite state
is lower bounded by the sum of local coherence in each subsystem and the genuine multipartite entanglement in the multipartite system.
Moreover, we obtain the monogamy of coherence in terms of IQ coherence measure. Therefore, the IQ coherence measures introduced here
truly capture the non-sharability of quantumness.

\section{main result}
Let $\cH$ be a d-dimensional Hilbert space and $\cD(\cH)$ the set
of density operators acting on $\cH$.
The max-relative entropy of coherence  for a given state $\rho\in \cD(\cH)$ is defined as
\begin{eqnarray*}
C_{\max}(\rho)
=\min_{\sigma\in\cI}
D_{\max}(\rho||\sigma),
\end{eqnarray*}
where max-relative entropy $D_{\max}$ \cite{Datta2009IEEE,Datta2009} is defined as
\begin{eqnarray*}
D_{\max}(\rho||\sigma)
=\min\set{\lambda\in \real_+:\rho\leq2^\lambda \sigma}.
\end{eqnarray*}
The coherence measure $C_{\max}$ has been proved to play an crucial role in some quantum information processing tasks in Ref. \cite{Bu2017b}.
For multipartite state $\rho\in\cD(\cH^{\ot N})$, $C_{\max}(\rho)=\min_{\sigma_N\in\cI_{1:2:\cdots:N}}D_{\max}(\rho||\sigma_N)$, where the incoherent states
$\cI_{1:2:\cdots:N}$ is the set of incoherent states $\cD(\cH^{\ot N})$ and the state $\sigma_N\in \cI_{1:2:\cdots:N}$ has the following form
\begin{eqnarray}
\sigma_N=\sum_i p_i \sigma_{i,1}\ot\ldots\ot\sigma_{i,N},
\end{eqnarray}
with all $\sigma_{i,k}$ being diagonal in the local basis.

To quantify the coherence in multipartite system more precisely, let us introduce the following coherence measure on bipartite system, which is called max-relative entropy of incoherent-quantum (IQ) coherence. For a bipartite
state $\rho_{AB}\in\cD(\cH_A\ot \cH_B)$,
\begin{eqnarray}
C^{A|B}_{\max}(\rho_{AB})
:=\min_{\sigma_{A|B}\in \cI Q}D_{\max}(\rho_{AB}||\sigma_{A|B}),
\end{eqnarray}
where the set of incoherent-quantum  states $\cI Q$ \cite{Chitambar2016,StreltsovPRX2017} is given by
\begin{eqnarray}\nonumber
\nonumber\cI Q=\{\sigma_{A|B}\in\cD(\cH_A\ot \cH_B)|\sigma_{A|B}=\sum_ip_i\sigma^A_i\ot \tau^B_i,\\ \nonumber
\sigma^A_i~ \text{is incoherent}, \tau^B_i\in\cD(\cH_B)\}.
\end{eqnarray}
As a coherence measure on bipartite systems, $C^{A|B}_{\max}$ satisfies the following properties:
(i) positivity, $C^{A|B}_{\max}(\rho)\geq 0$ and $C^{A|B}_{\max}(\rho)=0$ iff $\rho\in \cI Q$;
(ii) monotonicity under incoherent operation $\Lambda^A_{IO}$ on A side, that is, $C^{A|B}_{\max}(\Lambda^A_{IO}\ot \mathbb{I}(\rho_{AB}))\leq C^{A|B}_{\max}(\rho_{AB})$;
(iii) strong monotonicity under incoherent operation on A side, that is, for incoherent operation   $\Lambda^A_{IO}(\cdot)=\sum_i K^A_i(\cdot)K^{A\dag}_i$
with $K^A_i\cI K^{A\dag}_i\subset\cI$,
$\sum_i p_iC^{A|B}_{\max}(\tilde{\rho}_i)\leq C^{A|B}_{\max}(\rho)$,
where $p_i=\Tr{K^A_i\rho_{AB} K^{A\dag}_i}$ and $\tilde{\rho}_i=K^A_i\rho_{AB} K^{A\dag}_i/p_i$;
(iv) monotonicity under quantum operation $\Lambda^B$ on B side, $C^{A|B}_{\max}(\mathbb{I}\ot \Lambda^B(\rho_{AB}))\leq C^{A|B}_{\max}(\rho_{AB})$;
(v) quasi-convexity, for $\rho_{AB}=\sum^n_ip_i\rho_i$,
$C^{A|B}_{\max}(\rho_{AB})\leq\max_{i}C^{A|B}_{\max}(\rho_i)$.

The proof of above properties can be given in the same way as that of
$C_{\max}$ in Ref. \cite{Bu2017b}.
For any bipartite state $\rho_{AB}$ with
$\rho_A=\Ptr{B}{\rho_{AB}}$, $C^{A|B}_{\max}(\rho_{AB})\geq C_{\max}(\rho_A)$ , which comes directly from the property (iv). If the subsystem B is
trivial, i.e., $dim\cH_B=1$, then $C^{A|B}_{\max}$ reduces to $C_{\max}$ on subsystem $A$.

\textit{\textbf{Maximum advantage achievable in subchannel discrimination with the assistance of a quantum memory.}}-- In the following,
we investigate the information processing task: subchannel discrimination problem, which provides
an operational interpretation of  $C^{A|B}_{\max}$. Subchannel discrimination is an important information task and it tells
us which branch of the evolution a quantum system should go \cite{Piani2015PRL}.

A subchannel $\cE$  is defined to be a linear completely positive and trace non-increasing map, and if
the subchannel $\cE$ is also trace preserving, then $\cE$ is called a channel.
An instrument $\mathfrak{J}=\set{\cE_k}_k$ for a channel $\cE$ is a collection of subchannels $\cE_k$ such that
$\cE=\sum_k\cE_k$ \cite{Piani2015PRL}. An incoherent instrument $\mathfrak{I}^I$ for an IO $\cE$
is a collection of subchannels $\set{\cE_k}_k$ such that $\cE=\sum_k\cE_k$ \cite{Bu2017b}.

Given an bipartite state $\rho_{AB}$ and an instrument $\mathfrak{I}_A=\set{\cE^A_k}_k$ for a quantum channel $\cE^A$ on part A, consider the
joint positive operator valued measure (POVM) $\set{M^{AB}_k}_k$  on AB  with $M^{AB}_k\geq 0$ and $\sum_kM^{AB}_k=\mathbb{I}_{AB}$.
The probability of successfully discriminating subchannels in  $\mathfrak{I}_A$
by joint POVM $\set{M^{AB}_k}_k$ is given by
\begin{eqnarray*}
p_{\text{succ}}(\mathfrak{I}_A, \set{M^{AB}_k}_k, \rho_{AB})=\sum_k\Tr{\cE^A_k(\rho_{AB})M^{AB}_{k}}.
\end{eqnarray*}
And the optimal probability of successfully discriminating subchannels in  $\mathfrak{I}_A$ over all joint POVM is given by
\begin{eqnarray*}
p_{\text{succ}}(\mathfrak{I}_A, \rho_{AB})=\max_{\set{M^{AB}_k}_k}\sum_k\Tr{\cE^A_k(\rho_{AB})M^{AB}_{k}}.
\end{eqnarray*}
If the input states  are restricted to be incoherent on A's side, i.e., $\cI Q$ states, then the optimal probability over all $\cI Q$ states is
\begin{eqnarray*}
p^{IQ}_{\text{succ}}(\mathfrak{I}_A)=\max_{\sigma_{A|B}\in\cI Q} p_{\text{succ}}(\mathfrak{I}_A, \sigma_{A|B}).
\end{eqnarray*}

\begin{thm}\label{eq:op_mean}
Given a bipartite state $\rho_{AB}\in\cD(\cH_A\ot\cH_B)$, it holds that
\begin{eqnarray}
2^{C^{A|B}_{\max}(\rho_{AB})}=\max_{\mathfrak{I}^I_A}
\frac{p_{\text{succ}}(\mathfrak{I}^I_A,\rho_{AB})}{p^{IQ}_{\text{succ}}(\mathfrak{I}^I_A)},
\end{eqnarray}
where the maximization is taken over all the incoherent instrument $\mathfrak{I}^I_A$ on part A.
\end{thm}

The proof of Theorem \ref{eq:op_mean} is presented in Appendix \ref{apen:max_min}. This result illustrates that the maximal advantage of bipartite states in such subchannels discrimination problem
can be exactly captured by $C^{A|B}_{\max}$ , which also
provides an operational interpretation of $C^{A|B}_{\max}$.
As for any bipartite state $\rho_{AB}$ with reduced state $\rho_A$ on subsystem A, $C^{A|B}_{\max}(\rho_{AB})\geq
C_{\max}(\rho_A)$,
this means that  the success probability of discriminating subchanels on part A can be improved with the assistance of a quantum memory B. ( See Fig. \ref{fig2})

\begin{figure}
\centering
\includegraphics[width=70mm]{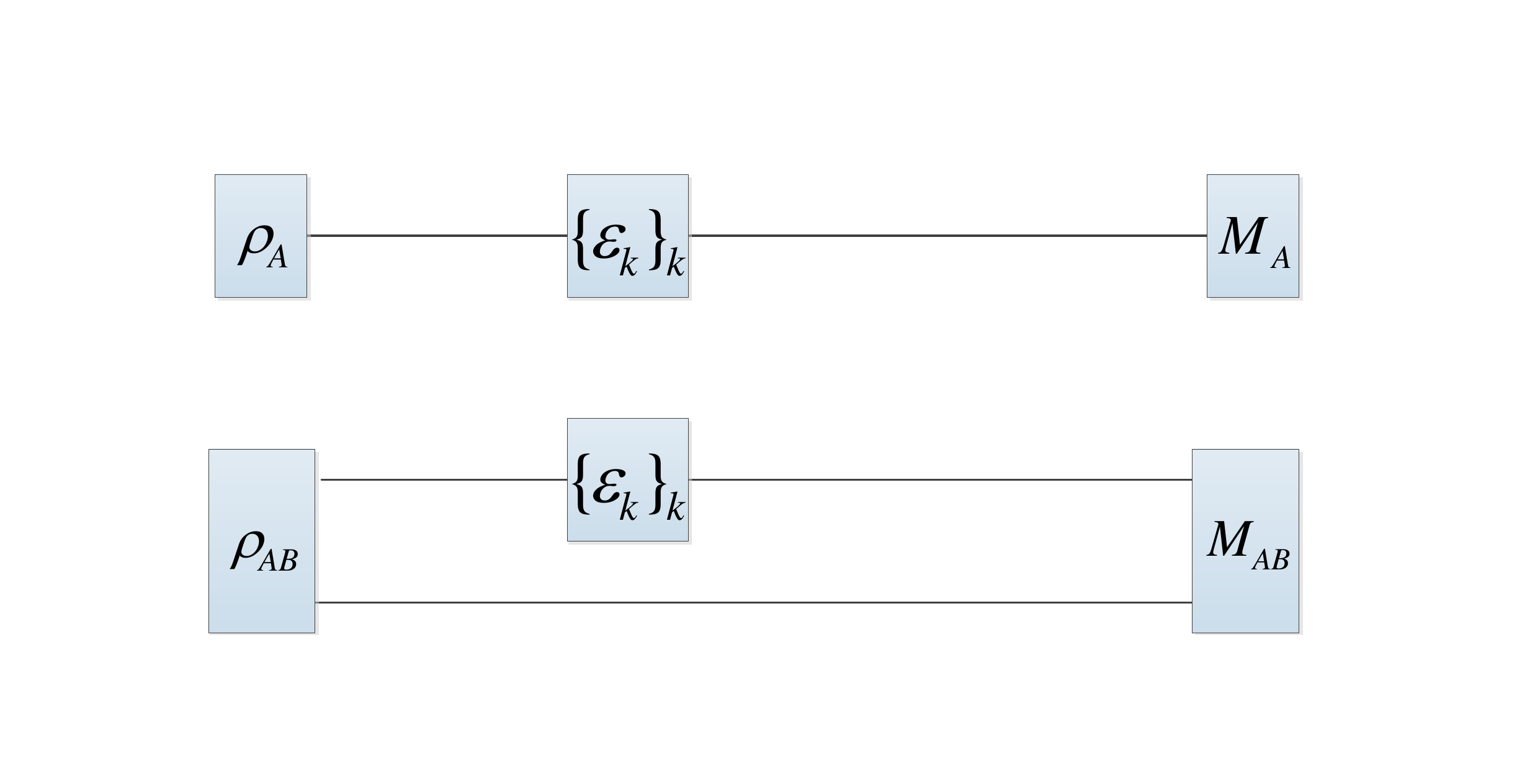}
 \caption{Schematic diagram of subchannels determination problem. In the first scenario, the subchannels are only distinguished by
 local measurement on A side, and the advantage achievable by coherent states is captured by $C_{\max}(\rho_{A})$ \cite{Bu2017b}. In the second scenario,
the subchannels are distinguished by the joint measurement on AB, and the advantage achievable by bipartite states is captured by $C^{A|B}_{\max}(\rho_{AB})$.
Here, $M_A$ (or $M_{AB}$) denotes the measurement on A (or AB).}
  \label{fig2}
\end{figure}

In Ref. \cite{Bu2017b}, the min-relative entropy of coherence $C_{\min}$ has also been defined,
\begin{eqnarray*}
C_{\min}(\rho)=\min_{\sigma\in\cI}D_{\min}(\rho||\sigma),
\end{eqnarray*}
where the min-relative entropy $D_{\min}$ \cite{Datta2009IEEE,Datta2009} is defined as
\begin{eqnarray*}
D_{\min}(\rho||\sigma):=-\log\Tr{\Pi_{\rho}\sigma},
\end{eqnarray*}
with $\Pi_{\rho}$ denoting the projector onto the support $\text{supp} [\rho]$ of $\rho$.
Here, we define the  min-relative entropy of IQ coherence on bipartite states,
\begin{eqnarray}
C^{A|B}_{\min}(\rho_{AB})=\min_{\sigma_{A|B}\in\cI Q}D_{\min}(\rho_{AB}||\sigma_{A|B}).
\end{eqnarray}
Moreover, the relative entropy of IQ coherence measure $C^{A|B}_r$ has also been defined in Ref. \cite{Chitambar2016},
\begin{eqnarray*}
C^{A|B}_{r}(\rho_{AB})
:=\min_{\sigma_{A|B}\in \cI Q}S(\rho_{AB}||\sigma_{A|B}),
\end{eqnarray*}
where $C^{A|B}_r$ plays an important role in the assisted distillation of coherence \cite{Chitambar2016,StreltsovPRX2017}.
Since $D_{\min}(\rho||\sigma)\leq S(\rho||\sigma)\leq D_{\max}(\rho||\sigma)$  for any quantum states $\rho$ and $\sigma$ \cite{Datta2009IEEE}, we have the following relationship,
\begin{eqnarray}
C^{A|B}_{\min}(\rho_{AB})
\leq C^{A|B}_r(\rho_{AB})
\leq C^{A|B}_{\max}(\rho_{AB}).
\end{eqnarray}

Let us introduce the  $\epsilon$-smooth
max- and min-relative entropy of  IQ  coherence as follows,
 \begin{eqnarray}
C^{A|B, \epsilon}_{\max}(\rho_{AB}):&=&\min_{\rho'_{AB}\in B_{\epsilon}(\rho_{AB})}C^{A|B}_{\max}(\rho'_{AB}),\\
C^{A|B,\epsilon}_{\min}(\rho_{AB}):&=&\max_{\substack{
0\leq O_{AB}\leq \mathbb{I}_{AB}\\
\Tr{O_{AB}\rho_{AB}}\geq 1-\epsilon}}
\min_{\sigma_{A|B}\in\cI Q}-\log\Tr{O_{AB}\sigma_{A|B}},~~~~~
\end{eqnarray}
where $B_\epsilon(\rho):=\set{\rho'\geq0:\norm{\rho'-\rho}_1\leq \epsilon, \Tr{\rho'}\leq\Tr{\rho}}$ and
 $\mathbb{I}_{AB}$ denotes the
identity on $\cH_A\ot\cH_B$. By the smooth max- and min-relative entropy of  IQ coherence, the equivalence between
$C^{A|B}_{\max}$, $C^{A|B}_{\min}$ and $C^{A|B}_{r}$ in the asymptotic limit can be also obtained,
\begin{eqnarray}
\lim_{\epsilon\to 0}\lim_{n\to \infty}
\frac{1}{n}C^{A^n|B^n,\epsilon}_{\max \setminus \min}(\rho^{\ot n}_{AB})=C^{A|B}_r(\rho_{AB}).
\end{eqnarray}
The proof of this result is presented in Appendix \ref{apen:prop_mami}.

Now, we are ready to investigate the coherence distribution in multipartite systems. For convenience, we denote
coherence measure $C(A|B):=C^{A|B}(\rho_{AB})$,
$C(AB):=C(\rho_{AB})$, $C^{\epsilon}(A|B):=C^{A|B,\epsilon}(\rho_{AB})$ and $C^{\epsilon}(AB):=C^{\epsilon}(\rho_{AB})$.
Although coherence
is defined to capture the quantumness in a single system, collective coherence between different subsystems needs to be considered in multipartite systems. To quantify the collective coherence between different subsystems, the local coherence needs to be omitted. Thus the minimization over
the incoherent states $\cI_{1:2:\cdots:N}$ needs to be relaxed to the separable states $S_{1:2:\cdots:N}$ \cite{Radhakrishnan2016}, where $\tau_N\in S_{1:2:\cdots:N}$ has the form
$\tau_N=\sum_ip_i\tau_{i,1}\ot \tau_{i,2}\ot\ldots\ot\tau_{i,N}$ with $\tau_{i,k}\in\cD(\cH_k)$.
For an N-partite state $\rho_{A_1\ldots A_N}$, the  max-relative entropy of collective coherence
is defined by
\begin{eqnarray}\nonumber
E^{A_1:\ldots:A_N}_{\max}(\rho_{A_1\ldots A_N})=\min_{\tau_N\in S_{1:2:\cdots:N}}
D_{\max}(\rho_{A_1\ldots A_N}||\tau_N),
\end{eqnarray}
which is the genuine multipartite entanglement among $A_1, \ldots, A_N$ in terms of the max-relative entropy.
The multipartite entanglement measure $E^{A_1:\ldots:A_N}(\rho_{A_1\ldots A_N})$ is denoted by $E(A_1:\ldots:A_N)$ for short in the following context.

Since the max-relative entropy fulfils the triangle inequality, i.e., $D_{\max}(\rho||\sigma)\leq D_{\max}(\rho||\tau)+D_{\max}(\tau||\sigma)$,
we have the following relation for any N-partite state $\rho_{A_1\ldots A_N}$ according to \cite{Radhakrishnan2016},
\begin{eqnarray*}
C_{\max}(A_1\ldots A_N)
\leq E_{\max}(A_1:\ldots :A_N)+C_{\max}(\sigma^{s}_{\min,N}),
\end{eqnarray*}
where $\sigma^{s}_{\min,N}$ is the optimal separable states in $S_{1:2:\cdots:N}$ such that
$E_{\max}(A_1:\ldots: A_N)=D_{\max}(\rho_{A_1\ldots A_N}||\sigma^{s}_{\min,N})$, and the coherence in the
state $\sigma^{s}_{\min,N}$ is called ``local coherence " in Ref. \cite{Radhakrishnan2016}. Besides,
$C_{\max}$ is subadditive for product states, i.e., $C_{\max}(\rho_1\ot\rho_2)\leq C_{\max}(\rho_1)+C_{\max}(\rho_2)$. That is, $C_{\max}$ satisfies
all the requirements except for the symmetry in Ref. \cite{Radhakrishnan2016}. However, the relation between $C(\sigma^{s}_{\min,N})$ and the coherence
$\sum_{k}C(\rho_{A_k})$ is still unclear, where $\rho_{A_k}$  is the reduced state of the k-th subsystem.
We adopt the $C(\rho_{A_k})$ ( or $C(A_k)$ ) to be the local coherence on the k-th subsystem and concentrate on the relation among the total coherence of multipartite state $C(A_1\ldots A_N)$,
the local coherence $\set{C(A_k)}_k$ and the genuine multipartite entanglement $E(A_1:\ldots:A_N)$.

\textit{\textbf{Distribution of coherence in bipartite systems.}}--
Let us begin with bipartite systems, for which we have the following result for any quantum state $\rho_{AB}\in\cD(\cH_A\ot\cH_B)$,
\begin{eqnarray}\label{eq:bi_re}
C_r(AB)\geq C_r(A|B)+C_r(B).
\end{eqnarray}
The proof of this result is based on the following fact,
\begin{eqnarray*}
C^{\epsilon}_{\max}(AB)\geq C^{\epsilon'}_{\max}(A|B)+C^{\epsilon}_{\min}(B),
\end{eqnarray*}
where $\epsilon>0$, $\epsilon'=\epsilon+2\sqrt{\epsilon}$, $C^{\epsilon}_{\max}$ and $C^{\epsilon}_{\min}$ are smooth max- and min-relative entropy of coherence defined in Ref. \cite{Bu2017b}.
The details of the proof is presented in Appendix \ref{appen:pro_dis}.
As $C_r(A|B)\geq C_r(A)$, the relation \eqref{eq:bi_re} is
tighter than the known result  $C_r(AB)\geq C_r(A)+C_r(B)$ in bipartite systems.
This is because $C_r(A|B)$ (or $C_{\max}(A|B)$) contains not only the local coherence on part A, but also the nonlocal correlation
between A and B. In fact, $C_r(A|B)=C_r(A)+S(\rho_A)+\sum_ip_iS(\rho_{B|i})-S(\rho_{AB})\geq C_r(A)+\delta_{A\to B}$, where
$p_i=\Innerm{i}{\rho_A}{i}_A$, $\rho_{B|i}=\Innerm{i}{\rho_{AB}}{i}_A$ and
$\delta_{A\to B}:=S(\rho_A)-S(\rho_{AB})+\min_{\set{\pi^A_i}}\sum_ip_iS(\rho_{B|i})$ is the quantum discord between A and B \cite{Ollivier2001} with $\set{\pi^A_i}$ being the von Neumann measurements on part A.
Thus, it is easy to get the relation : $C_r(AB)\geq C_r(A)+C_r(B)+\delta_{A\to B}$. Similar result can also be obtained under the exchange of labels A and B. Therefore, we obtain the following relation for the
distribution of relative entropy of coherence in bipartite systems,
\begin{equation}
C_r(AB)\geq C_r(A)+C_r(B)+\max\set{\delta_{A\to B},\delta_{B\to A}},
\end{equation}
where $\delta_{A\to B}$ and  $\delta_{B\to A}$ are the corresponding quantum discord of bipartite state $\rho_{AB}$.

\textit{\textbf{Distribution of coherence in multipartite systems.}}--
In multipartite systems, the total coherence of the whole system also contains the nonlocal correlations between the subsystems.
\begin{thm}\label{thm:dis_coh_mul}
Given an N-partite state $\rho_{A_1A_2\ldots A_N}\in\cD(\ot^{N}_{k=1}\cH_{A_k})$ with $N\geq 2$, the total coherence, local coherence and genuine
multipartite entanglement of this state have the following relationship,
\begin{eqnarray}\nonumber
C_r(A_1A_2\ldots A_N)
\geq E^{\infty}_{r}(A_1:A_2:\ldots:A_N)+\sum^N_{i=1}C_r(A_i),\\
\end{eqnarray}
where $E^{\infty}_{r}$ is the regularized relative entropy of entanglement
defined by  $ E^{\infty}_{r}(A_1:A_2:\ldots:A_N)=\lim_{n\to \infty}\frac{1}{n}E_r(\rho^{\ot n}_{A_1:A_2:\ldots:A_N})$ with
$E_r(\rho_{A_1:A_2:\ldots: A_N})=\min_{\tau_N\in S_{1:2:\ldots:N}}S(\rho_{A_1:A_2:\ldots: A_N}||\tau_N)$ \cite{Vedral1997,Vedral1998,Vollbrecht2001}, the minimization goes over all separable states $\tau_N\in S_{1:2:\ldots:N}$.
\end{thm}

To prove  Theorem \ref{thm:dis_coh_mul}, we only need to prove the case $N=3$, which depends on the following relation,
\begin{eqnarray*}
C^{\epsilon}_{\max}(ABC)
&\geq& E^{\epsilon_2}_{\max}(A:B:C)+C^{\epsilon_1}_{\min}(A)\\
&&+C^{\epsilon_1}_{\min}(B)+C^{\epsilon}_{\min}(C),
\end{eqnarray*}
where $\epsilon_1=\epsilon+2\sqrt{\epsilon}$,
$\epsilon_2=\epsilon_1+2\sqrt{2\epsilon_1}$ and $E^{\epsilon}_{\max}$ is the smooth max-relative entanglement defined in Refs. \cite{Datta2009,Brandao2011}.
The details of the proof for  Theorem \ref{thm:dis_coh_mul} is presented in Appendix  \ref{appen:pro_dis}. This theorem illustrates that the total coherence in multipartite system contains not only the
local coherence in each subsystem, but also the genuine multipartite entanglement among the multipartite systems,  where the
multipartite entanglement quantifies the collective coherence among these subsystems.

Now, we consider the distribution of entanglement (or collective coherence) in multipartite system.
Although the relative entropy of entanglement and its regularized version do not have the monogamy relation in general \cite{Winter2016Entdis}, the genuine
multipartite entanglement for any tripartite state $\rho_{ABC}\in\cD(\cH_A\ot\cH_B\ot\cH_C)$ can be decomposed into bipartite entanglement as follows,
\begin{eqnarray}
\label{eq:ent1}E^{\infty}_r(A:B:C)&\geq& E^{\infty}_{r}(A:BC)+E^{\infty}_r(B:C),\\
\label{eq:ent2}E^{\infty}_r(A:B:C)&\geq& E^{\infty}_{r}(A:B)+E^{\infty}_r(B:C).
\end{eqnarray}
The proof of this result is presented in Appendix \ref{appen:pro_dis}.
Note that the relation \eqref{eq:ent1} is also true for any N-partite systems, i.e.,
\begin{eqnarray*}
E^{\infty}_r(A_1:\ldots:A_{N-1}: A_N) & \geq & E^{\infty}_{r}(A_1:\ldots:A_{N-1}A_N)\\
&& +E^{\infty}_r(A_{N-1}:A_N).
\end{eqnarray*}
That is, the genuine N-partite entanglement can be decomposed into
the $(N-1)$-partite entanglement and bipartite entanglement.
Moreover, since
the relation \eqref{eq:ent2} holds under the exchange of labels A, B and C, we have the
following relation for the distribution of  entanglement in tripartite systems:
\begin{eqnarray*}
E^{\infty}_r(A:B:C) &\geq & \frac{2}{3}[E^{\infty}_r(A:B)\\
&& +E^{\infty}_r(A:C)+E^{\infty}_r(B:C)].
\end{eqnarray*}


\textit{\textbf{Monogamy relation for IQ coherence measures.}}--
It has been shown in Refs. \cite{Yao2015,Kumar2017} that,
for relative entropy of coherence $C_r$, there exists some tripartite $\rho_{ABC}$
such that
\begin{eqnarray}\label{eq:nmon_r}
C_r(ABC)\leq C_r(AB)+C_r(AC).
\end{eqnarray}
That is, the monogamy relation for $C_r$ does not hold in general.
There are several reasons behind the failure of monogamy relation for $C_r$. One is the following fact:
the right hand side of \eqref{eq:nmon_r} contains two copies of local coherence $C_r(A)$, while the left hand side only contains one copy of $C_r(A)$. If the parts $B$ and $C$
are weakly correlated ( e.g. $\rho_{ABC}=\rho_{A_1B}\ot\rho_{A_2C}$ with $\cH_{A_1}\ot\cH_{A_2}=\cH_A$), then the one more copy of $C_r(A)$ will result in the failure of monogamy.
Thus, to circumvent this problem,  we define the monogamy of coherence for an $N+1$-part state $\rho_{A_1\ldots A_N B}$ in terms of the IQ coherence measure $C^{A|B}_r$ as follows,
\begin{eqnarray}
M=\sum^N_{k=1}C_r(A_k|B)-C_r(A_1\ldots A_N|B).
\end{eqnarray}
It is monogamous for $M\leq0$ and polymonogamous for $M>0$.
Here,  we obtain the following monogamy of coherence in $N+1$-partite systems.

\begin{thm}\label{thm:mon_coh}
For any N+1-partite state $\rho_{A_1\ldots A_N B}$, it holds that
\begin{eqnarray}\label{eq:mon_coh1}
C_r(A_1A_2\ldots A_N|B)\geq \sum^N_{i=1}C_r(A_i|B).
\end{eqnarray}
\end{thm}
The proof of Theorem \ref{thm:mon_coh} is presented in Appendix \ref{apen:mon_coh}.
As $C_r(AB|C)\geq C_r(AB)+E^{\infty}_r(AB:C)\geq C_r(A)+C_r(B)+E^{\infty}_r(A:B)+E^{\infty}_r(AB:C)$, i.e., $C_r(AB|C)$ contains the nonlocal correlation between
$A$ and $B$,  we consider the relation between the quantity $C_r(AB|C)-C_r(A|C)-C_r(B|C)$
and the nonlocal correlation between A and B.
For any tripartite state $\rho_{ABC}$,  the
following relationship holds,
\begin{equation}\label{eq:re_condi}
C_r(AB|C)
\leq C_r(A|C)+C_r(B|C)+I(A:B|C)_{\rho},
\end{equation}
where
$I(A:B|C)_{\rho}:=S(\rho_{AC})+S(\rho_{BC})-S(\rho_C)-S(\rho_{ABC})$ is the conditional mutual information of $\rho_{ABC}$ which quantifies the correlation between subsystems A and B with respect to C and  can be used to define the squashed entanglement \cite{Christandl2004,Tucci2002,BrandaoCMP2011} (See the proof in  Appendix \ref{apen:mon_coh}).

One of the basic properties of the relative entropy of entanglement $E_r$ distinguishing from other entanglement measures is the non-lockability \cite{HorodeckiPRL2005},
that is,  the loss of  entanglement is proportional to
the number of qubits traced out when some part of the whole system is discarded,
where this relation  can be improved for the regularized relative entropy of entanglement $E^{\infty}_r$ \cite{BrandaoCMP2011} as follows,
\begin{eqnarray}\nonumber
E^{\infty}_r(A:BC)- E^{\infty}_r(A:C)\leq I(A:B|C)_{\rho}.
\end{eqnarray}
For the relative entropy of IQ coherence measure, the conditional mutual information $I(A:B|C)_{\rho}$
 also provides an upper bound for the loss of coherence after some  subsystem is discarded. That is,  for any tripartite state $\rho_{ABC}$,
\begin{eqnarray}\label{eq:cond_r_bi}
C_r(A|BC)-C_r(A|C)\leq I(A:B|C)_{\rho},
\end{eqnarray}
 which can be regarded as the non-lockability of relative entropy of IQ coherence measure. The proof of
 \eqref{eq:cond_r_bi} is presented in Appendix \ref{apen:mon_coh}.

\section{Conclusion}
Understanding the distribution of quantum coherence in multipartite systems is of fundamental importance.
We have investigated the distribution of coherence in multipartite systems by introducing incoherent-quantum (IQ) coherence measures on bipartite systems in terms of the max- and min-relative entropies. It has been found that the max-relative entropy of IQ coherence characterizes maximal advantage of the bipartite in certain subchannel discrimination problems. Moreover, it has been shown
that the total coherence of a multipartite system is lower bounded by the sum of local coherence and the genuine multipartite entanglement. From the IQ coherence measures, we have obtained the monogamy relation of coherence in multipartite systems.

Our results reveal the distribution of quantum coherence in multipartite systems,
which substantially advance the understanding of the physical law that governs the distribution of quantum correlations in
multipartite systems and pave the way for the further researches in this direction. This will also have deep implications in quantum information processing, quantum biology, quantum thermodynamics
and other related areas of physics as well.

\medskip
\begin{acknowledgments}
J.D. Wu is supported by the Natural Science Foundation of China (Grants No. 11171301, No. 10771191, and No. 11571307) and the Doctoral Programs Foundation of the Ministry
of Education of China (Grant No. J20130061).
S.M. Fei is supported by the Natural Science Foundation of China under No. 11675113.
\end{acknowledgments}

\bibliographystyle{apsrev4-1}
\bibliography{Maxcoh-lit}
\appendix
\section{Operational interpretation of  $C^{A|B}_{\max}$}\label{apen:max_min}
\begin{lem}\label{lem:semd_form}
Given a bipartite quantum state $\rho_{AB}\in\cD(\cH_{A}\ot\cH_B)$, $C^{A|B}_{\max}(\rho_{AB})$ can be expressed as
\begin{eqnarray}
2^{C^{A|B}_{\max}(\rho_{AB})}=\max_{\substack{
\tau_{AB}\geq0\\
\Delta_{A}\ot \mathbb{I}_B(\tau_{AB})=\mathbb{I}_{AB}
}}\Tr{\rho_{AB}\tau_{AB}}.
\end{eqnarray}
\end{lem}
\begin{proof}
Due to the definition of $C^{A|B}_{\max}$, we have $2^{C^{A|B}_{\max}(\rho_{AB})}=\min_{\substack{
\sigma_{AB}\geq0\\
\Delta_A\ot \mathbb{I}_B(\sigma_{AB})\geq \rho_{AB}
}}\Tr{\sigma_{AB}}$. Thus, to prove the result, we only need to prove
\begin{eqnarray*}
\min_{\substack{
\sigma_{AB}\geq0\\
\Delta_A\ot \mathbb{I}_B(\sigma_{AB})\geq \rho_{AB}
}}\Tr{\sigma_{AB}}=\max_{\substack{
\tau_{AB}\geq0\\
\Delta_{A}\ot \mathbb{I}_B(\tau_{AB})=\mathbb{I}_{AB}
}}\Tr{\rho_{AB}\tau_{AB}}.
\end{eqnarray*}

First, it is easy to see that
\begin{eqnarray*}
\max_{\substack{
\tau_{AB}\geq0\\
\Delta_A\ot \mathbb{I}_B(\tau_{AB})=\mathbb{I}_{AB}
}}\Tr{\rho_{AB}\tau_{AB}}
=\max_{\substack{
\tau_{AB}\geq0\\
\Delta_A\ot \mathbb{I}_B(\tau_{AB})\leq \mathbb{I}_{AB}
}}\Tr{\rho_{AB}\tau_{AB}},
\end{eqnarray*}
as for any positive operator $\tau_{AB}\geq0$ with $\Delta_A(\tau_{AB})\leq \mathbb{I}_{AB}$, we can always choose $\tau'_{AB}=\tau_{AB}+\mathbb{I}_{AB}-\Delta_A\ot \mathbb{I}_B(\tau_{AB})\geq0$, then  $\Delta_A\ot \mathbb{I}_B(\tau_{AB})=\mathbb{I}_{AB}$ and
$\Tr{\rho_{AB}\tau'_{AB}}\geq \Tr{\rho_{AB}\tau_{AB}}$.

Next, we prove that
\begin{eqnarray}\label{eq:SDP1}
\min_{\substack{
\sigma_{AB}\geq0\\
\Delta_A\ot \mathbb{I}_B(\sigma_{AB})\geq \rho_{AB}
}}\Tr{\sigma_{AB}}=\max_{\substack{
\tau\geq0\\
\Delta_A\ot \mathbb{I}_B(\tau_{AB})\leq \mathbb{I}_{AB}
}}\Tr{\rho_{AB}\tau_{AB}}.~~~~~~~~
\end{eqnarray}
The left side of equation \eqref{eq:SDP1} can be expressed as the following semidefinite
programming (SDP)
\begin{eqnarray*}
\min \Tr{C_1\sigma_{AB}},\\
\text{s.t.}~~ \Lambda(\sigma_{AB})\geq C_2,\\
\sigma_{AB}\geq 0,
\end{eqnarray*}
where $C_1=\mathbb{I}$, $C_2=\rho_{AB}$ and $\Lambda=\Delta_{A}\ot \mathbb{I}_B$. Then the dual SDP is given by
\begin{eqnarray*}
\max \Tr{C_2\tau_{AB}},\\
\text{s.t.}~~ \Lambda^\dag(\tau_{AB})\leq C_1,\\
\tau_{AB}\geq 0.
\end{eqnarray*}
That is,
\begin{eqnarray*}
\max \Tr{\rho_{AB}\tau_{AB}},\\
\text{s.t.}~~ \Delta_A\ot \mathbb{I}_B(\tau_{AB})\leq \mathbb{I}_{AB},\\
\tau_{AB}\geq 0.
\end{eqnarray*}
Note that the dual is strictly feasible as we only need to choose $\sigma_{AB}=2\lambda_{\max}(\rho_{AB}) \mathbb{I}_{AB}$, where $\lambda_{\max}(\rho_{AB})$ is the
 maximum eigenvalue of $\rho_{AB}$. Thus, strong duality holds, and the equation \eqref{eq:SDP1} is proved.

\end{proof}

\begin{mproof}[Proof of Theorem \ref{eq:op_mean}]
Due to the definition of $C^{A|B}_{\max}$, there exists an $\cI Q$ state $\sigma_{A|B}$
such that $\rho\leq 2^{C^{A|B}_{\max}(\rho_{AB})}\sigma_{A|B}$. Thus, for any
incoherent instrument $\mathfrak{I}^I_A$ and joint
POVM $\set{M^{AB}_k}_k$,
\begin{eqnarray*}
&&p_{\text{succ}}(\mathfrak{I}^I_A,\set{M^{AB}_k}_k,\rho_{AB})\\
&\leq& 2^{C^{A|B}_{\max}(\rho_{AB})} p_{\text{succ}}(\mathfrak{I}^I_A, \set{M^{AB}_k}_k, \sigma_{A|B}),
\end{eqnarray*}
that is,
\begin{eqnarray}\label{eq:dis_eql}
p_{\text{succ}}(\mathfrak{I}^I_A,\rho_{AB})\leq 2^{C^{A|B}_{\max}(\rho_{AB})} p^{IQ}_{\text{succ}}(\mathfrak{I}^I_A).
\end{eqnarray}

Now, consider a special incoherent instrument such that the equality \eqref{eq:dis_eql} holds.
Let us take the incoherent instrument  $\widetilde{\mathfrak{I}}^I_A=\set{\widetilde{\cE}^A_k}_k$
as $\widetilde{\cE}^A_k(\cdot)=\frac{1}{d_A}U^A_k(\cdot) U^{A\dag}_k$ with $U^A_k=e^{i\frac{2k\pi}{d_A}H_A}$ and $H_A=\sum_{j}j\proj{j}_A$.
Then, for any $\cI Q$ state $\sigma_{A|B}$,
\begin{eqnarray*}
&&p_{\text{succ}}(\widetilde{\mathfrak{I}}^I_A, \set{M^{AB}_k},\sigma_{A|B})\\
&=&\sum_{k}\frac{1}{d_A}\Tr{U^A_k\sigma_{A|B}U^{A\dag}_kM^{AB}_k}\\
&=&\frac{1}{d_A}\Tr{\sigma_{A|B}\sum_kU^{A\dag}_kM^{AB}_kU^A_k}\\
&=&\frac{1}{d_A}\Tr{\Delta_A\ot \mathbb{I}_B(\sigma_{A|B})\sum_kU^{A\dag}_kM^{AB}_kU^A_k}\\
&=&\frac{1}{d_A}\Tr{\sigma_{AB}\Delta_A\ot \mathbb{I}_B(\sum_kU^{A\dag}_kM^{AB}_kU^A_k)}\\
&=&\frac{1}{d_A}\Tr{\sigma_{AB}}
=\frac{1}{d_A},
\end{eqnarray*}
where the forth line comes from the fact that $\Delta_A\ot \mathbb{I}_B(\sigma_{A|B})=\sigma_{A|B}$  for any $\cI Q$ state $\sigma_{A|B}$,
the sixth line comes from the fact that $U^A_k$ is diagonal in the local basis $\set{\proj{j}_A}^{d_A}_{j=1}$ and
\begin{eqnarray*}
&&\Delta_A\ot \mathbb{I}_B(\sum_kU^{A\dag}_kM^{AB}_kU^A_k)\\
&=&\sum_k U^{A\dag}_k\Delta_A\ot I_B(M^{AB}_k)U^A_k \\
&=&\sum_k\Delta_A\ot \mathbb{I}_B(M^{AB}_k)\\
&=&\Delta_A\ot \mathbb{I}_B(\sum_k M^{AB}_k)\\
&=&\Delta_A\ot \mathbb{I}_B(\mathbb{I}_{AB})=\mathbb{I}_{AB}.
\end{eqnarray*}

Besides, according to Lemma \ref{lem:semd_form}, there exists an positive operator $\tau_{AB}$ with $\Delta_A\ot \mathbb{I}_B(\tau_{AB})=\mathbb{I}_{AB}$,
such that $2^{C^{A|B}_{\max}(\rho_{AB})}=\Tr{\rho_{AB}\tau_{AB}}$. Define $N^{AB}_k$ to be
$N^{AB}_{k}=\frac{1}{d_A}U^{A\dag}_k\tau_{AB}U^A_k$, then $N^{AB}_{k}\geq 0$ and $\sum_kN^{AB}_k=\frac{1}{d_A}\sum_kU^{A\dag}_k\tau_{AB}U^A_k=\Delta_A\ot \mathbb{I}_B(\tau_{AB})=\mathbb{I}_{AB}$, that is, $\set{N^{AB}_{k}}$ is a joint POVM on A and B. Hence
\begin{eqnarray*}
&&p_{\text{succ}}(\widetilde{\mathfrak{I}}^I_A, \set{N^{AB}_k},\rho_{AB})\\
&=&\sum_k\frac{1}{d_A}\Tr{U^A_k\rho_{AB}U^{A\dag}_kN^{AB}_k}\\
&=&\sum_k\frac{1}{d^2_A}\Tr{\rho_{AB}\tau_{AB}}\\
&=&\frac{1}{d_A}\Tr{\rho_{AB}\tau_{AB}}\\
&=&\frac{2^{C^{A|B}_{\max}(\rho_{AB})}}{d_A}\\
&=&2^{C^{A|B}_{\max}(\rho_{AB})}p^{IQ}_{\text{succ}}(\widetilde{\mathfrak{I}}^I_A).
\end{eqnarray*}
That is,
\begin{eqnarray}
\frac{p_{\text{succ}}(\widetilde{\mathfrak{I}}^I_A,\rho_{AB})}{p^{IQ}_{\text{succ}}(\widetilde{\mathfrak{I}}^I_A)}
\geq2^{C^{A|B}_{\max}(\rho_{AB})}.
\end{eqnarray}

\end{mproof}

\section{Properties of smooth IQ max- and min- relative entropies}\label{apen:prop_mami}
Due to the definition of  smooth max-relative entropy of IQ coherence measure, it can also be rewritten as
\begin{eqnarray*}
C^{A|B,\epsilon}_{\max}(\rho_{AB})&=&\min_{\rho'_{AB}\in B_{\epsilon}(\rho_{AB})}\min_{\sigma_{A|B}\in \cI Q}D_{\max}(\rho'_{AB}||\sigma_{A|B})\\
&=&\min_{\sigma_{A|B}\in\cI Q} D^{\epsilon}_{\max}(\rho_{AB}||\sigma_{A|B}),
\end{eqnarray*}
where $D^{\epsilon}_{\max}(\rho||\sigma)$  is the smooth max-relative entropy \cite{Datta2009IEEE,Datta2009,Brandao2011} defined as
\begin{eqnarray}\label{eq:s_max}
D^\epsilon_{\max}(\rho||\sigma)=\inf_{\rho'\in B_{\epsilon}(\rho)}D_{\min}(\rho'||\sigma),
\end{eqnarray}
and $B_\epsilon(\rho):=\set{\rho'\geq0:\norm{\rho'-\rho}_1\leq \epsilon, \Tr{\rho'}\leq\Tr{\rho}}$. The
equivalence between $C^{A|B,\epsilon}_{\max}$ and $C^{A|B}_r$ in the asymptotic limit is given in the following proposition.

\begin{prop}\label{eq:r_vs_max}
Given a bipartite state $\rho_{AB}\in\cD(\cH_A\ot \cH_B)$, we have
\begin{eqnarray}
C^{A|B}_r(\rho_{AB})=\lim_{\epsilon\to 0}\lim_{n\to \infty}
\frac{1}{n}C^{A^n|B^n,\epsilon}_{\max}(\rho^{\ot n}_{AB}).
\end{eqnarray}
\end{prop}

\begin{proof}
The set of incoherent-quantum states $\cI_{A^n}Q_{B^n}$ in $\cD((\cH_A\ot\cH_B)^{\ot n})$ has the
following form $\sigma_{A^n|B^n}=\sum_ip_i\sigma^{A_1}_{1}\ot\ldots\ot \sigma^{A_n}_{i}\ot \tau^{B_1\ldots B_n}_i$ with
$\sigma^{A_k}_i$ being incoherent and  $\tau^{B_1\ldots B_n}_i\in\cD(\cH^{\ot n}_B)$ for any $i,k$. It is easy to see that
the family of sets $\set{\cI_{A^n}Q_{B^n}}_n$ with $\cI_{A^n}Q_{B^n}\in \cD((\cH_A\ot \cH_B)^{\ot n})$ satisfy the five conditions required in Ref. \cite{Brandao2010CMP} as follows:
(1) Each $\cI_{A^n}Q_{B^n}$ is convex and closed;
(2) Each $\cI_{A^n}Q_{B^n}$ contains a state $\sigma^{\ot n}$ with $\sigma\in \cD(\cH_A\ot\cH_B)$ being full rank;
(3) If $\rho\in \cI_{A^{n+1}}Q_{B^{n+1}}$, then $\Ptr{k}{\rho}\in\cI_{A^n}Q_{B^n}$ for any $k\in \set{1, ..., n+1}$ where
$\mathrm{Tr}_k$ means the partial trace on the $k$th $\cH_A\ot\cH_B$ in $(\cH_A\ot\cH_B)^{\ot n}$;
(4) If $\rho\in \cI_{A^n}Q_{B^n}$ and $\tau\in \cI_{A^m}Q_{B^m}$, then $\rho\ot \tau\in \cI_{A^{m+n}}Q_{B^{m+n}}$;
(5) Each $\cI_{A^n}Q_{B^n}$ is permutation invariant, that is, for every state $\rho\in \cI_{A^n}Q_{B^n}$, 
$P_{\pi}\rho P_{\pi}\in \cI_{A^n}Q_{B^n}$ where $S_n$ is the symmetry group group of order $n$ and $P_{\pi}$ is the
representation of the element $\pi\in S_n$ in the space $(\cH_A\ot\cH_B)^{\ot n}$ which is given by 
$P_{\pi}(\psi_1\ot ... \ot \psi_n)=\psi_{\pi^{-1}(1)}\ot....\ot \psi_{\pi^{-1}(n)}$.
Thus according to the generalized Quantum Stein Lemma \cite{Brandao2010CMP},
we have
\begin{eqnarray*}
&&\lim_{\epsilon\to 0}\lim_{n\to \infty}
\frac{1}{n}C^{A^n|B^n,\epsilon}_{\max}(\rho^{\ot n}_{AB})\\
&=&\lim_{n\to \infty}\min_{\sigma_{A^n|B^n}\in \cI_{A^n}Q_{B^n}}\frac{S(\rho^{\ot n}_{AB}||\sigma_{A^n|B^n})}{n}.
\end{eqnarray*}
Since
\begin{eqnarray*}
&&\min_{\sigma_{A^n|B^n}\in \cI_{A^n}Q_{B^n}}S(\rho^{\ot n}_{AB}||\sigma_{A^n|B^n})\\
&=&S(\Delta^{\ot n}_A(\rho^{\ot n}_{AB}))-S(\rho^{\ot n}_{AB})\\
&=&n[S(\Delta_A(\rho_{AB}))-S(\rho_{AB})]\\
&=&nC^{A|B}_r(\rho_{AB}),
\end{eqnarray*}
we get the result.

\end{proof}

In view of the definition of smooth min-relative entropy of  IQ coherence measure, it can be expressed as follows,
\begin{eqnarray*}
C^{A|B,\epsilon}_{\min}(\rho_{AB})
=\min_{\sigma_{A|B}\in\cI Q}D^{\epsilon}_{\min}(\rho_{AB}||\sigma_{A|B}),
\end{eqnarray*}
where $D^{\epsilon}_{\min}(\rho||\sigma)$ is the smooth min-relative entropy \cite{Datta2009IEEE,Datta2009,Brandao2011}
defined as
\begin{eqnarray}\label{eq:s_min}
D^\epsilon_{\min}(\rho||\sigma)=
\sup_{\substack{
0\leq O\leq \mathbb{I}\\
\Tr{O\rho}\geq 1-\epsilon}}
-\log\Tr{O\sigma}.
\end{eqnarray} First,
since the set of incoherent-quantum states $\cI_{A^n}Q_{B^n}$ in $\cD((\cH_A\ot\cH_B)^{\ot n})$ satisfy the conditions
 in \cite{Brandao2010CMP}, we have the following lemma.

\begin{lem}\label{lem:stein}
Given a quantum state $\rho_{AB}\in\cD(\cH_A\ot \cH_B)$,

(Direct part) For any $\epsilon>0$, there exists a sequence of POVMs $\set{M_{A^nB^n}, \mathbb{I}-M_{A^nB^n}}_n$ such that
\begin{eqnarray}
\lim_{n\to\infty}
\Tr{(\mathbb{I}-M_{A^nB^n})\rho^{\ot n}_{AB}}=0,
\end{eqnarray}
and for every integer $n$ and every  incoherent-quantum state $\sigma_{A^n|B^n}\in\cI_{A^n}Q_{B^n}$,
\begin{eqnarray}
-\frac{\log\Tr{M_{A^nB^n}\sigma_{A^n|B^n}}}{n}+\epsilon
\geq C^{A|B}_r(\rho_{AB}).
\end{eqnarray}

(Strong converse) If there exists $\epsilon>0$ and a sequence of POVMs $\set{M_{A^nB^n}, \mathbb{I}-M_{A^nB^n}}_n$ such that for every integer $n>0$ and $\sigma_{A^n|B^n}\in\cI_{A^n}Q_{B^n}$,
\begin{eqnarray}
-\frac{\log\Tr{M_{A^nB^n}\sigma_{A^n|B^n}}}{n}-\epsilon
\geq C^{A|B}_r(\rho_{AB}),
\end{eqnarray}
then
\begin{eqnarray}
\lim_{n\to\infty}\Tr{(\mathbb{I}-M_{A^nB^n})\rho^{\ot n}_{AB}}=1.
\end{eqnarray}

\end{lem}

\begin{prop}\label{eq:r_vs_max}
Given a bipartite state $\rho_{AB}\in\cD(\cH_A\ot \cH_B)$, we have
\begin{eqnarray}
C^{A|B}_r(\rho_{AB})=\lim_{\epsilon\to 0}\lim_{n\to \infty}
\frac{1}{n}C^{A^n|B^n,\epsilon}_{\min}(\rho^{\ot n}_{AB}).
\end{eqnarray}
\end{prop}

\begin{proof}
Due to the fact that $C^{A|B}_r(\rho_{AB})=\lim_{\epsilon\to 0}\lim_{n\to \infty}
\frac{1}{n}C^{A^n|B^n,\epsilon}_{\max}(\rho^{\ot n}_{AB})$ and
$D^{\epsilon}_{\min}\leq D^{\epsilon}_{\max}-\log(1-2\epsilon)$ \cite{Bu2017b},
we have
\begin{eqnarray}
\lim_{\epsilon\to 0}\lim_{n\to \infty}
\frac{1}{n}C^{A^n|B^n,\epsilon}_{\min}(\rho^{\ot n}_{AB})\leq C^{A|B}_r(\rho_{AB}).
\end{eqnarray}

Next, we  prove the converse direction,
\begin{eqnarray}
\lim_{\epsilon\to 0}\lim_{n\to \infty}
\frac{1}{n}C^{A^n|B^n,\epsilon}_{\min}(\rho^{\ot n}_{AB})\geq C^{A|B}_r(\rho_{AB}).
\end{eqnarray}
According to Lemma \ref{lem:stein},
for any $\epsilon>0$, there exists a sequence of POVMs $\set{M_{A^nB^n}}$ such that for sufficient large integer $n$,  $\Tr{\rho^{\ot n}_{AB}M_{A^nB^n}}\geq1-\epsilon$. Hence
\begin{eqnarray*}
C^{A^n|B^n,\epsilon}_{\min}(\rho^{\ot n}_{AB})
&\geq& \min_{\sigma_{A^n|B^n}\in\cI^nQ^n}-\log\Tr{M_{A^nB^n}\sigma_{A^n|B^n}}\\
&\geq& n(C^{A|B}_r(\rho)-\epsilon),
\end{eqnarray*}
where the last inequality comes from the direct part of Lemma \ref{lem:stein}.
Therefore,
\begin{eqnarray*}
\lim_{\epsilon\to 0}
\lim_{n\to\infty}
\frac{1}{n}C^{A^n|B^n,\epsilon}_{\min}(\rho^{\ot n}_{AB})
\geq C^{A|B}_r(\rho_{AB}).
\end{eqnarray*}
\end{proof}

\section{Details about the proof in the distribution of coherence }\label{appen:pro_dis}
To prove the results in the distribution of coherence, the following lemmas is necessary.

\begin{lem}\label{lem:gen_op}
( \text{Gentle operator lemma \cite{Winter1999,Ogawa2007}})
Suppose $\rho$ is a subnormalized state with $\Tr{\rho}\leq1$, $\rho\geq 0$ and $M$
is an operator with $0\leq M\leq \mathbb{I}$,
$\Tr{\rho M}\geq \Tr{\rho}-\epsilon$, then
\begin{eqnarray}
\norm{\rho-\sqrt{M}\rho\sqrt{M}}_1\leq 2\sqrt{\epsilon},
\end{eqnarray}
where $\norm{A}_1=\Tr{\sqrt{A^\dag A}}$.
\end{lem}

\begin{lem}\label{lem:1}
Given a bipartite state $\rho_{AB}\in\cD(\cH_A\ot\cH_B)$ and a parameter $\epsilon\geq 0$,
\begin{eqnarray}
C^{\epsilon}_{\max}(AB)\geq C^{\epsilon'}_{\max}(A|B)+C^{\epsilon}_{\min}(B),
\end{eqnarray}
where $\epsilon'=\epsilon+2\sqrt{\epsilon}$, $C^{\epsilon}_{\max}$ and $C^{\epsilon}_{\min}$ are smooth max- and min-relative entropy of coherence defined in \cite{Bu2017b},
\begin{eqnarray*}
C^{\epsilon}_{\max}(\rho):&=&\min_{\rho'\in B_{\epsilon}(\rho)}C_{\max}(\rho'),\\
C^{\epsilon}_{\min}(\rho):&=&\max_{\substack{
0\leq O\leq \mathbb{I}\\
\Tr{O\rho}\geq 1-\epsilon}}
\min_{\sigma\in\cI}-\log\Tr{O\sigma},
\end{eqnarray*}
with $B_\epsilon(\rho)=\set{\rho'\geq0:\norm{\rho'-\rho}_1\leq \epsilon, \Tr{\rho'}\leq\Tr{\rho}}$.
\end{lem}
\begin{proof}
Due to the definition of
$C^{\epsilon}_{\max}$, there exists an optimal state  $\bar{\rho}_{AB}\in B_{\epsilon}(\rho_{AB})$ such that
$C^{\epsilon}_{\max}(AB)=C_{\max}(\bar{\rho}_{AB})$. Let us take
$\lambda=2^{C^{\epsilon}_{\max}(AB)}$, then there exists an incoherent state $\sigma_{AB}=\sum_{i}p_i\sigma_{A,i}\ot \sigma_{B,i}\in\cI_{AB}$ with
$\sigma_{A,i}, \sigma_{B,i}$ being incoherent such that
\begin{eqnarray}\label{ineq:bi_max}
\lambda\sum_ip_i\sigma_{A,i}\ot \sigma_{B,i}\geq \bar{\rho}_{AB}.
\end{eqnarray}
According to the definition of $C^{\epsilon}_{\min}(\rho_B)$ with $\rho_B=\Ptr{A}{\rho_{AB}}$, there exists an positive operator $T_{B}$ with $0\leq T_B\leq \mathbb{I}_B$, $\Tr{T_B\rho_B}\geq 1-\epsilon$ such that $C^{\epsilon}_{\min}(B)=\min_{\sigma_B\in\cI_B}-\log\Tr{T_B\sigma_B}$. Thus, for any incoherent state $\sigma_B\in\cI_B$, $2^{-C^{\epsilon}_{\min}(\rho_B)}\geq \Tr{T_B\sigma_B}$.

 Applying $\sqrt{T_B}(\cdot) \sqrt{T_B}$ on the both sides of \eqref{ineq:bi_max}, we get
\begin{eqnarray}
\lambda\sum_ip_i\sigma_{A,i}\ot \sqrt{T_B}\sigma_{B,i}\sqrt{T_B}\geq \sqrt{T_B}\bar{\rho}_{AB}\sqrt{T_B}.
\end{eqnarray}
Then $\widetilde{\sigma}_{AB}=\sum_ip_i\sigma_{A,i}\ot \sqrt{T_B}\sigma_{B,i}\sqrt{T_B}=:\mu\sigma_{AB}$ with $\sigma_{AB}\in \cI_AQ_B$ and
$\mu=\Tr{\widetilde{\sigma}_{AB}}=\sum_ip_i\Tr{T_B\sigma_{B,i}}\leq 2^{-C^{\epsilon}_{\min}(B)}$.
Besides, $\rho'_{AB}=\sqrt{T_B}\bar{\rho}_{AB}\sqrt{T_B}\in B_{\epsilon'}(\rho_{AB})$ with $\epsilon'=\epsilon+2\sqrt{\epsilon}$ as
\begin{eqnarray*}
&&\norm{\sqrt{T_B}\bar{\rho}_{AB}\sqrt{T_B}-\rho_{AB}}_1\\
&\leq& \norm{\sqrt{T_B}\bar{\rho}_{AB}\sqrt{T_B}-\sqrt{T_B}\rho_{AB}\sqrt{T_B}}_1\\
&&+\norm{\sqrt{T_B}\rho_{AB}\sqrt{T_B}-\rho_{AB}}_1\\
&\leq& \norm{\bar{\rho}_{AB}-\rho_{AB}}_1+\norm{\sqrt{T_B}\rho_{AB}\sqrt{T_B}-\rho_{AB}}_1\\
&\leq& \epsilon+2\sqrt{\epsilon},
\end{eqnarray*}
where the last inequality comes from the fact that $\bar{\rho}_{AB}\in B_{\epsilon}(\rho_{AB})$ and the gentle operator lemma \ref{lem:gen_op}.
Hence, $\lambda\mu \sigma_{A|B}\geq \rho'_{AB}$ with $\sigma_{A|B}\in \cI_A Q_B$.
Thus,
\begin{eqnarray*}
C^{\epsilon'}_{\max}(A|B)
&\leq& C^{A|B}_{\max}(\rho'_{AB})\\
&\leq& \log\lambda+\log\mu\\
&\leq& C^{\epsilon}_{\max}(AB)-C^{\epsilon}_{\min}(B).\\
\end{eqnarray*}

\end{proof}

\begin{mproof}[Proof of Eq. \eqref{eq:bi_re}]
This result comes directly from the Lemma \ref{lem:1} and the following facts,
\begin{eqnarray}
\label{eq:smooth1}\lim_{\epsilon\to 0}\lim_{n\to \infty}
\frac{1}{n}C^{\epsilon}_{\min \setminus \max}(\rho^{\ot n})&=&C_r(\rho),\\
\label{eq:smooth2}\lim_{\epsilon\to 0}\lim_{n\to \infty}
\frac{1}{n}C^{A^n|B^n,\epsilon}_{\min \setminus \max}(\rho^{\ot n}_{AB})&=&C^{A|B}_r(\rho_{AB}),
\end{eqnarray}
where \eqref{eq:smooth1} has been proved in Ref. \cite{Bu2017b}.
\end{mproof}

\begin{lem}\label{lem:coh_ent}
Given a bipartite $\rho_{AB}\in\cD(\cH_A\ot\cH_B)$ and a parameter $\epsilon\geq 0$, we have
\begin{eqnarray}
C^{\epsilon}_{\max}(A|B)
\geq E^{\epsilon'}_{\max}(A:B)+C^{\epsilon}_{\min}(A),
\end{eqnarray}
where  $\epsilon'=\epsilon+2\sqrt{\epsilon}$ and $E^{\epsilon'}_{\max}$ is the smooth max-relative entropy of entanglement defined in \cite{Datta2009,Brandao2011} as follows,
\begin{eqnarray*}
E^{\epsilon}_{\max}(A:B):=\min_{\rho'_{AB}\in B_{\epsilon}(\rho_{AB})}E_{\max}(\rho'_{AB}),
\end{eqnarray*}
with $B_\epsilon(\rho)=\set{\rho'\geq0:\norm{\rho'-\rho}_1\leq \epsilon, \Tr{\rho'}\leq\Tr{\rho}}$.
Moreover, for any tripartite state $\rho_{ABC}\in\cD(\cH_A\ot\cH_B\ot\cH_C)$, we have
\begin{eqnarray}\nonumber
C^{\epsilon}_{\max}(AB|C)
\geq E^{\epsilon''}_{\max}(A:B:C)+C^{\epsilon}_{\min}(A)+C^{\epsilon}_{\min}(B),
\end{eqnarray}
where  $\epsilon''=\epsilon+2\sqrt{2\epsilon}$.
\end{lem}
\begin{proof}
The proof of the tripartite is similar to that of  bipartite case. Hence, we only prove the bipartite case.
Due to the definition of
$C^{\epsilon}_{\max}(A|B)$, there exists an optimal state $\bar{\rho}_{AB}\in B_{\epsilon}(\rho_{AB})$ such that
$C^{\epsilon}_{\max}(A|B)=C^{A|B}_{\max}(\bar{\rho}_{AB})$. Let us take
$\lambda=2^{C^{\epsilon}_{\max}(A|B)}$, then there exists an incoherent-quantum state $\sigma_{A|B}=\sum_{i}p_i\sigma_{A,i}\ot \tau_{B,i}\in\cI_{A}Q_B$ with
$\sigma_{A,i}$ being incoherent and $\tau_{B,i}\in\cD(\cH)$ such that
\begin{eqnarray}\label{ineq:bi_max2}
\lambda\sum_ip_i\sigma_{A,i}\ot \tau_{B,i}\geq \bar{\rho}_{AB}.
\end{eqnarray}
According to the definition of $C^{\epsilon}_{\min}$, there exists an positive operator $T_{A}$ with $0\leq T_A\leq \mathbb{I}_A$, $\Tr{T_A\rho_A}\geq 1-\epsilon$ such that $C^{\epsilon}_{\min}(A)=\min_{\sigma_A\in\cI_A}-\log\Tr{T_A\sigma_A}$. Thus, for any incoherent state $\sigma_A\in\cI_A$,
$C^{\epsilon}_{\min}(A)\leq -\log\Tr{T_A\sigma_A}$, i.e., $2^{-C^{\epsilon}_{\min}(A)}\geq \Tr{T_A\sigma_A}$.

Applying $\sqrt{T_A}(\cdot) \sqrt{T_A}$ on the both sides of \eqref{ineq:bi_max2}, one gets
\begin{eqnarray}
\lambda\sum_ip_i\sqrt{T_A}\sigma_{A,i}\sqrt{T_A}\ot \tau_{B,i}\geq \sqrt{T_A}\bar{\rho}_{AB}\sqrt{T_A}.
\end{eqnarray}
Then $\widetilde{\sigma}_{AB}=\sum_ip_i\sqrt{T_A}\sigma_{A,i}\sqrt{T_A}\ot \tau_{B,i}=:\mu\sigma_{A:B}$ with $\sigma_{A:B}$ being separable, i.e.,  $\sigma_{A:B}\in S_{A:B}$ and
$\mu=\Tr{\widetilde{\sigma}_{AB}}=\sum_ip_i\Tr{T_A\sigma_{A,i}}\leq 2^{-C^{\epsilon}_{\min}(A)}$.
Besides, $\rho'_{AB}=\sqrt{T_A}\bar{\rho}_{AB}\sqrt{T_A}\in B_{\epsilon'}(\rho_{AB})$ with $\epsilon'=\epsilon+2\sqrt{\epsilon}$.
That is, $\lambda\mu \sigma_{A:B}\geq \rho'_{AB}$ with $\sigma_{A:B}\in S_{A:B}$.
Thus,
\begin{eqnarray*}
E^{\epsilon'}_{\max}(A:B)
&\leq& E_{\max}(\rho'_{AB})\\
&\leq& \log\lambda+\log\mu\\
&\leq& C^{\epsilon}_{\max}(A|B)-C^{\epsilon}_{\min}(A).\\
\end{eqnarray*}

\end{proof}

\begin{mproof}[ Proof of Theorem \ref{thm:dis_coh_mul}]
Here, we only need to prove the case N=3.
Due to Lemmas \ref{lem:1} and \ref{lem:coh_ent}, we have
\begin{eqnarray*}
C^{\epsilon}_{\max}(ABC)
&\geq& E^{\epsilon_2}_{\max}(A:B:C)+C^{\epsilon_1}_{\min}(A)\\
&&+C^{\epsilon_1}_{\min}(B)+C^{\epsilon}_{\min}(C),
\end{eqnarray*}
where $\epsilon_1=\epsilon+2\sqrt{\epsilon}$ and
$\epsilon_2=\epsilon_1+2\sqrt{2\epsilon_1}$.
It has been proved that $E_{\max}$ is equivalent to the regularized relative entropy of entanglement $E^{\infty}_r$ in asymptotic limit \cite{Datta2009,Brandao2011}, that is,
\begin{eqnarray}\nonumber
\lim_{\epsilon\to 0}\lim_{n\to \infty}
\frac{1}{n}E^{\epsilon}_{\max}(\rho^{\ot n}_{AB})=E^{\infty}_r(A:B).
\end{eqnarray}
Thus, combined with \eqref{eq:smooth1},  we obtain the  theorem.
\end{mproof}

\begin{lem}\label{lem:dis_ent}
Given a tripartite state $\rho_{ABC}\in\cD(\cH_A\ot\cH_B\ot\cH_C)$ and a parameter $\epsilon\geq 0$,
\begin{eqnarray}\nonumber
E^{\epsilon}_{\max}(A:B:C)&\geq& E^{\epsilon'}_{\max}(A:BC)+E^{\epsilon}_{\min}(B:C),\\ \nonumber
E^{\epsilon}_{\max}(A:B:C)&\geq& E^{\epsilon'}_{\max}(A:B)+E^{\epsilon}_{\min}(B:C),
\end{eqnarray}
where $\epsilon'=\epsilon+2\sqrt{\epsilon}$ and $E^{\epsilon'}_{\min}$ is the smooth min-relative entropy of entanglement defined in \cite{Datta2009,Brandao2011} as follows,
\begin{eqnarray*}
E^{\epsilon}_{\min}(B:C):=\max_{\substack{
0\leq O\leq \mathbb{I}\\
\Tr{O\rho_{BC}}\geq 1-\epsilon}}
\min_{\sigma\in S_{B:C}}-\log\Tr{O\sigma},
\end{eqnarray*}
with $S_{B:C}$ being  the set of separable states.
\end{lem}
\begin{proof}
The proof of these two inequalities are similar,  we only prove the first one.
Due to the definition of
$E^{\epsilon}_{\max}(A:B:C)$, there exists an optimal state  $\bar{\rho}_{ABC}\in B_{\epsilon}(\rho_{ABC})$ such that
$E^{\epsilon}_{\max}(A:B:C)=E^{A:B:C}_{\max}(\bar{\rho}_{ABC})$. Taking
$\lambda=2^{E^{\epsilon}_{\max}(A:B:C)}$, we have that there exists a separable state $\tau_{A:B:C}=\sum_{i}p_i\tau_{A,i}\ot \tau_{B,i} \ot \tau_{C,i}\in S_{A:B:C}$ such that
\begin{eqnarray}\label{ineq:bi_max3}
\lambda\sum_ip_i\tau_{A,i}\ot \tau_{B,i} \ot \tau_{C,i}\geq \bar{\rho}_{ABC}.
\end{eqnarray}
According to the definition of $E^{\epsilon}_{\min}(B:C)$, there exists an positive operator $T_{BC}$ with $0\leq T_{BC}\leq \mathbb{I}_{BC}$, $\Tr{T_{BC}\rho_{BC}}\geq 1-\epsilon$ such that $E^{\epsilon}_{\min}(B:C)=\min_{\tau_{B:C}\in S_{B:C}}-\log\Tr{T_{BC}\tau_{B:C}}$. Thus, for any separable $\tau_{B:C}\in S_{B:C}$,
$E^{\epsilon}_{\min}(B:C)\leq -\log\Tr{T_{BC}\tau_{B:C}}$, i.e.,
$E2^{-E^{\epsilon}_{\min}(B:C)}\leq \Tr{T_{BC}\tau_{B:C}}$.

Now  applying $\sqrt{T_{BC}}(\cdot) \sqrt{T_{BC}}$ on the both sides of \eqref{ineq:bi_max3}, one obtains
\begin{eqnarray*}
&&\lambda\sum_ip_i\sqrt{T_{BC}}\tau_{A,i}\ot \tau_{B,i} \ot \tau_{C,i}\sqrt{T_{BC}}\\
&\geq& \sqrt{T_{BC}}\bar{\rho}_{ABC}\sqrt{T_{BC}}.
\end{eqnarray*}
Then $\widetilde{\tau}_{A:BC}=\sum_ip_i\sqrt{T_{BC}}\tau_{A,i}\ot \tau_{B,i} \ot \tau_{C,i}\sqrt{T_{BC}}=:\mu\tau_{A:BC}$, where $\tau_{A:BC}$ is separable with respect to the partition $A:BC$, i.e.,  $\tau_{A:BC}\in S_{A:BC}$ and
$\mu=\Tr{\widetilde{\tau}_{A:BC}}=\sum_ip_i\Tr{T_{BC}\tau_{B,i}\ot \tau_{C,i}}\leq 2^{-E^{\epsilon}_{\min}(B:C)}$.
Besides, $\rho'_{AB}=\sqrt{T_{BC}}\bar{\rho}_{ABC}\sqrt{T_{BC}}\in B_{\epsilon'}(\rho_{AB})$ with $\epsilon'=\epsilon+2\sqrt{\epsilon}$.
That is, $\lambda\mu \tau_{A:BC}\geq \rho'_{ABC}$ with $\tau_{A:BC}\in S_{A:BC}$.
Thus,
\begin{eqnarray*}
E^{\epsilon'}_{\max}(A:BC)
&\leq& E^{A:BC}_{\max}(\rho'_{ABC})\\
&\leq& \log\lambda+\log\mu\\
&\leq& E^{\epsilon}_{\max}(A:B:C)-E^{\epsilon}_{\min}(B:C).\\
\end{eqnarray*}

\end{proof}

\begin{mproof}[ Proof of Eqs. (\ref{eq:ent1}) and (\ref{eq:ent2})]
This result comes directly from Lemma \ref{lem:dis_ent}
and the equivalence between $E^{\epsilon}_{\max}$ and $E^{\infty}_r$ in the asymptotic limit.
\end{mproof}

\section{Monogamy of coherence}\label{apen:mon_coh}
\begin{lem}\label{lem:mon_coh}
Given a tripartite state $\rho_{ABC}\in\cD(\cH_A\ot\cH_B\ot\cH_C)$, one has
\begin{eqnarray}
\label{eq:chain_rule}
C^{\epsilon}_{\max}(AB|C)
&\geq& C^{\epsilon'}_{\max}(A|BC)+C^{\epsilon}_{\min}(B|C),\\
\label{eq:mon_coh}
C^{\epsilon}_{\max}(AB|C)
&\geq& C^{\epsilon'}_{\max}(A|C)+C^{\epsilon}_{\min}(B|C),
\end{eqnarray}
where $\epsilon'=\epsilon+2\sqrt{\epsilon}$.
\end{lem}
\begin{proof}
The proof of these two inequalities is similar,  we only prove the second one.
Due to the definition of
$C^{\epsilon}_{\max}(AB|C)$, there exists an optimal state $\bar{\rho}_{ABC}\in B_{\epsilon}(\rho_{ABC})$ such that
$C^{\epsilon}_{\max}(AB|C)=C^{AB|C}_{\max}(\bar{\rho}_{ABC})$. Let us take
$\lambda=2^{C^{\epsilon}_{\max}(AB|C)}$, then there exists an incoherent-quantum state $\sigma_{AB|C}=\sum_{i}p_i\sigma_{A,i}\ot \sigma_{B,i} \ot \tau_{C,i}\in \cI_{AB}Q_C$ with $\sigma_{A,i}, \sigma_{B,i}$ being incoherent such that
\begin{eqnarray}\label{ineq:bi_max3}
\lambda\sum_ip_i\sigma_{A,i}\ot \sigma_{B,i} \ot \tau_{C,i}\geq \bar{\rho}_{ABC}.
\end{eqnarray}
According to the definition of $C^{\epsilon}_{\min}(B|C)$, there exists a positive operator $T_{BC}$ with $0\leq T_{BC}\leq \mathbb{I}_{BC}$, $\Tr{T_{BC}\rho_{BC}}\geq 1-\epsilon$ such that $C^{\epsilon}_{\min}(B|C)=\min_{\sigma_{B|C}\in \cI_BQ_C}-\log\Tr{T_{BC}\sigma_{B|C}}$. Thus, for any in incoherent-quantum state $\sigma_{B|C}\in \cI_BQ_C$,
$C^{\epsilon}_{\min}(B|C)\leq -\log\Tr{T_{BC}\sigma_{B|C}}$, i.e.,
$2^{-C^{\epsilon}_{\min}(B|C)}\geq \Tr{T_{BC}\sigma_{B|C}}$.

Applying $\sqrt{T_{BC}}(\cdot) \sqrt{T_{BC}}$ on the both sides of \eqref{ineq:bi_max3} and take a partial trace on part B, we get
\begin{eqnarray*}
&&\lambda\sum_ip_i\Ptr{B}{\sqrt{T_{BC}}\sigma_{A,i}\ot\sigma_{B,i} \ot \tau_{C,i}\sqrt{T_{BC}}}\\
&&\geq \Ptr{B}{\sqrt{T_{BC}}\bar{\rho}_{ABC}\sqrt{T_{BC}}}.
\end{eqnarray*}
Then $\widetilde{\sigma}_{A|C}=\sum_ip_i\Ptr{B}{\sqrt{T_{BC}}\sigma_{A,i}\ot\sigma_{B,i} \ot \tau_{C,i}\sqrt{T_{BC}}}=:\mu\sigma_{A|C}$ where $\sigma_{A|C}$ is an incoherent-quantum state, i.e.,  $\sigma_{A|C}\in \cI_AQ_C$ and
$\mu=\Tr{\widetilde{\sigma}_{A|C}}=\sum_ip_i\Tr{T_{BC}\sigma_{B,i}\ot \tau_{C,i}}\leq 2^{-C^{\epsilon}_{\min}(B|C)}$.
Besides, $\rho'_{AC}=\Ptr{B}{\sqrt{T_{BC}}\bar{\rho}_{ABC}\sqrt{T_{BC}}}\in B_{\epsilon'}(\rho_{AC})$ with $\epsilon'=\epsilon+2\sqrt{\epsilon}$, as
$
\norm{\rho'_{AC}-\rho_{AC}}_1\leq \norm{\sqrt{T_{BC}}\bar{\rho}_{ABC}\sqrt{T_{BC}}-\rho_{ABC}}_1\leq \epsilon'
$.
That is, $\lambda\mu \sigma_{A|C}\geq \rho'_{AC}$ with $\sigma_{A|C}\in \cI_AQ_C$.
Thus,
\begin{eqnarray*}
C^{\epsilon'}_{\max}(A|C)
&\leq& C^{A|C}_{\max}(\rho'_{AC})\\
&\leq& \log\lambda+\log\mu\\
&\leq& C^{\epsilon}_{\max}(AB|C)-C^{\epsilon}_{\min}(B|C).\\
\end{eqnarray*}

\end{proof}

\begin{mproof}[ Proof of Theorem \ref{thm:mon_coh}]
These results come directly from Lemma \ref{lem:mon_coh} and the equivalence between
$C^{A|B,\epsilon}_{\max\setminus\min}$ and $C^{A|B}_r$ in the asymptotic limit.
\end{mproof}

In the data processing, there is a chain rule for von Neumann entropy, $S(AB|C)=S(A|BC)+S(B|C)$, where the condition entropy $S(A|B)$ is defined by $S(A|B):=S(\rho_{AB})-S(\rho_{B})$ with $\rho_B$ being the reduced state on subsystem B. There is also a chain rule for $C_r$ as follows:
for any tripartite $\rho_{ABC}\in\cD(\cH_A\ot\cH_B\ot\cH_C)$, it holds that
\begin{eqnarray}\label{eq:chain_cr}
C_r(AB|C)\geq C_r(A|BC)+C_r(B|C),
\end{eqnarray}
which results from Lemma \ref{lem:mon_coh}  by the equivalence between
$C^{A|B,\epsilon}_{\max\setminus\min}$ and $C^{A|B}_r$ in the asymptotic limit.

\begin{mproof}[Proof of Equation \eqref{eq:re_condi}]
The result comes directly from the definition of coherence measure $C_r$.
Due to the definitions, $C_r(AB|C)$, $C_r(A|C)$ and $C_r(B|C)$ can be written as
\begin{eqnarray*}
C_r(AB|C)&=&S(\Delta_A\ot\Delta_B(\rho_{ABC}))-S(\rho_{ABC});\\
C_r(A|C)&=&S(\Delta_A(\rho_{AC}))-S(\rho_{AC});\\
C_r(B|C)&=&S(\Delta_B(\rho_{BC}))-S(\rho_{BC});\\
\end{eqnarray*}
where $\rho_{AC}$ and $\rho_{BC}$ are the corresponding reduced states of $\rho_{ABC}$.
Thus
\begin{eqnarray*}
&&C_r(AB|C)-C_r(A|C)-C_r(B|C)\\
&=&S(\Delta_A\ot\Delta_B(\rho_{ABC}))\\
&&-S(\rho_{ABC})-\Pa{S(\Delta_A(\rho_{AC}))-S(\rho_{AC})}\\
&&-\Pa{S(\Delta_B(\rho_{BC}))-S(\rho_{BC})}\\
&=&\Pa{S(\rho_{AC})+S(\rho_{BC})-S(\rho_{C})-S(\rho_{ABC})}\\
&&-[S(\Delta_A(\rho_{AC}))+S(\Delta_B(\rho_{BC}))\\
&&-S(\rho_{C})-S(\Delta_A\ot\Delta_B(\rho_{ABC}))]\\
&=&I(A:B|C)_{\rho}-I(A:B|C)_{\Delta_A\ot\Delta_B(\rho)}\\
&\leq &I(A:B|C)_{\rho},
\end{eqnarray*}
where the last inequality comes from the fact that $I(A:B|C)\geq 0$.

\end{mproof}

\begin{mproof}[Proof of Equation \eqref{eq:cond_r_bi}]
 This comes directly from  \eqref{eq:chain_cr} and Equation (18) in the main context.
\end{mproof}

\end{document}